\documentclass[conference, letterpaper]{IEEEtran}

\usepackage{amsfonts, amsmath, amssymb, amsthm}
\usepackage{bbm,cite, mathtools}
\usepackage{comment,enumitem}
\newtheorem{theorem}{Theorem}
\newtheorem{proposition}{Proposition}
\newtheorem{lemma}{Lemma}

\newtheorem{corollary}{Corollary}

\theoremstyle{definition}
\newtheorem{defn}{Definition}
\theoremstyle{remark}
\newtheorem*{remark}{Remark}
\newtheorem*{example}{Example}
\usepackage{tikz}
\usepackage{pgfplots}
\pgfplotsset{compat=1.14}
\usepackage[
font=small,labelfont=bf
]{caption}

\usepackage{xcolor}
\newif\ifcomment
\commenttrue

\newcommand{\eqdef}{\vcentcolon=}

\newcommand{\X}{\ensuremath{\mathbf{X}}}
\renewcommand{\d}{\ensuremath{\mathbf{d}}}
\renewcommand{\i}{\ensuremath{\mathbf{i}}}
\newcommand{\y}{\ensuremath{\mathbf{y}}}

\renewcommand{\c}{\ensuremath{\mathbf{c}}}

\renewcommand{\v}{\ensuremath{\mathbf{v}}}

\newcommand{\p}{\ensuremath{\mathbf{p}}}
\renewcommand{\a}{\ensuremath{\mathbf{a}}}
\renewcommand{\b}{\ensuremath{\mathbf{b}}}
\newcommand{\0}{\ensuremath{\mathbf{0}}}

\newcommand{\F}{\ensuremath{\mathbb{F}}}
\newcommand{\cF}{\ensuremath{\mathcal{F}}}
\newcommand{\C}{\ensuremath{\mathcal{C}}}
\renewcommand{\L}{\ensuremath{\mathcal{L}}}
\newcommand{\Z}{\ensuremath{\mathbb{Z}}}
\renewcommand{\epsilon}{\varepsilon}

\newcommand{\Mod}[1]{\ (\mathrm{mod}\ #1)}
\newcommand{\Mods}[1]{\ (\mathrm{mod}^*\ #1)}

\usepackage{array}
\newcolumntype{C}{>{$}c<{$}} 

\begin{document}
\title{Lifted Reed-Solomon Codes with Application to Batch Codes}
\author{%
    \IEEEauthorblockN{\textbf{Lukas~Holzbaur}\IEEEauthorrefmark{1}, \textbf{Rina~Polyanskaya}\IEEEauthorrefmark{2},
\textbf{Nikita~Polyanskii}\IEEEauthorrefmark{1}\IEEEauthorrefmark{3}, and \textbf{Ilya~Vorobyev}\IEEEauthorrefmark{3}\IEEEauthorrefmark{4} }
	\IEEEauthorblockA{\IEEEauthorrefmark{1}Technical University of Munich, Germany}
		\IEEEauthorblockA{\IEEEauthorrefmark{2}Institute for Information Transmission Problems, Russia }
	\IEEEauthorblockA{\IEEEauthorrefmark{3}
		Skolkovo Institute of Science and Technology, Russia}
	\IEEEauthorblockA{\IEEEauthorrefmark{4}
		Moscow Institute of Physics and Technology,  Russia}
	\IEEEauthorblockA{\textbf{Emails}: lukas.holzbaur@tum.de, rev-rina@yandex.ru, nikita.polyansky@gmail.com, vorobyev.i.v@yandex.ru}
	\thanks{L. Holzbaur's work was supported by the Technical University of Munich -- Institute for Advanced Study, funded by the German Excellence Initiative and European Union 7th Framework Programme under Grant Agreement No. 291763 and the German Research Foundation (Deutsche Forschungsgemeinschaft, DFG) under Grant No. WA3907/1-1. Rina Polyanskaya and Ilya Vorobyev were supported in part by the Russian Foundation for Basic Research through grant no.~\mbox{20-01-00559}. N. Polianskii's research was supported in part by a German Israeli Project Cooperation (DIP) grant under grant no.~KR3517/9-1.}
}
\IEEEoverridecommandlockouts
\maketitle
\begin{abstract}
	Guo, Kopparty and Sudan have initiated the study of error-correcting codes derived by lifting of affine-invariant codes. Lifted Reed-Solomon (RS) codes are defined as the evaluation of polynomials in a vector space over a field by requiring  their restriction to every line in the space to be a codeword of the RS code. In this paper, we investigate lifted RS codes and discuss their application to batch codes, a notion introduced in the context of private information retrieval and load-balancing in distributed storage systems. First, we improve the estimate of the code rate of lifted RS codes for lifting parameter $m\ge 3$ and large field size. Second, a new explicit construction of batch codes utilizing lifted RS codes is proposed. For some parameter regimes, our codes have a better trade-off between parameters than previously known batch codes.
	\end{abstract}
	\begin{IEEEkeywords}
		Lifting, batch codes, Reed-Solomon codes, distributed storage systems, disjoint recovering sets
	\end{IEEEkeywords}
\section{Introduction}
Modern distributed storage systems are commonly set up to provide a large number of users access to the data, where each user is free to request any file stored in the system. To avoid delays and bottlenecks in data delivery, it is desirable for the system being able to serve each set of requested files by distributing the load, i.e., the task of transmitting some of its stored data to the user, among the servers in the system. While replicating all files on each servers allows for a trivial manner of balancing this load, it entails a large storage overhead. On the other hand, the use of classical erasure codes, such as Reed-Solomon (RS) codes, can minimize this overhead, but generally doesn't provide an efficient method of load balancing. \emph{Batch codes} are a class of codes which aim to bridge this gap.
\subsection{Related work}
Batch codes were originally motivated by different applications such as load-balancing in storage and cryptographic protocols~\cite{ishai2004batch}.  Several explicit and non-explicit constructions of these codes have been proposed, employing methods based on generalizations of Reed-Muller (RM) codes \cite{ishai2004batch,polyanskaya2019batch}, unbalanced expanders \cite{ishai2004batch}, graph theory \cite{rawat2016batch}, array and multiplicity codes \cite{asi2018nearly}, and finite geometries~\cite{polyanskii2019constructions}. In this work, we consider a special notion of batch codes, namely \emph{primitive multiset batch codes} (for a more general study on the different notions of batch codes the reader is referred to \cite{skachek2018batch}).

Informally, a primitive multiset $k$-batch code (in what follows, we simply write a $k$-batch code to refer to this class of codes) of length $N$ and dimension $n$ allows for the recovery of any set of $k$ message symbols, possibly with repetition, in $k$ disjoint ways, i.e., for any $k$-tuple (\textit{batch}) of message symbols $x_{i_1},...,x_{i_k}$ with $i_1,...,i_k \in [n]$ there exist $k$ non-intersecting sets $R_1,...,R_k \subset [N]$ such that the message symbol $x_{i_j}$ can be recovered from the codeword symbols indexed by the set $R_j$. 
 For large $k=\Omega(n)$, batch codes are closely related to constant-query \textit{locally correctable codes} and it is known~\cite{katz2000efficiency,woodruff2012quadratic} that their rate approaches zero. On the other hand, when $k=O(1)$ is fixed, there exist explicit code constructions with the code rate very close to one~\cite{vardy2016constructions}. 
 
 Because of the above motivation, we classify batch codes by the required redundancy $r(n,k)\eqdef N-n$. In this paper, we will be concerned with the regime of sublinear $k$, i.e., $k=n^\epsilon$ with ${n \to \infty}$ and $0\le\epsilon\le 1$. We write $\epsilon^-$ if a statement holds for any $\epsilon^*$ with $0\leq \epsilon^* < \epsilon$. Several achievability results, i.e., upper bounds on the smallest achievable $r(n,k)$, have been shown. We summarize the results that provide the smallest $r(n,n^\epsilon)$ for the binary batch codes and some $\epsilon$:
\begin{itemize}
    \item[\cite{polyanskaya2019batch}] $r(n,n^{\epsilon^-}) = O(n^{\log_4(3)+(2-\log_2(3))\epsilon})$ for $0<\epsilon<\frac{1}{2}$,
    \item[\cite{asi2018nearly}] $r(n,n^{\epsilon^-}) =  O(n^{g(\epsilon)})$ for $0\leq \epsilon \leq 1$, where 
\begin{equation*}
    g(\epsilon) \eqdef \min_{b \in \mathbb{N}: b > \frac{2}{1-\epsilon}} \left[ 1- \frac{b(1-\epsilon)-2}{4b(b-1)} \right] \ ,
\end{equation*}
    \item[\cite{polyanskii2019constructions}] $r(n,n^{\epsilon^-}) = O(n^{\frac{3\epsilon+1}{2}})$ for  $0<\epsilon<1/3$.
\end{itemize} 
On the other hand, the only non-trivial converse bound on the redundancy, yielding that $r(n,3)=\Omega(\sqrt{n})$, was obtained  independently in~\cite{wootters2016linear} and~\cite{rao2016lower} for linear \textit{private information retrieval codes} and for \textit{codes with the disjoint repair group property}, concepts closely related to batch codes.
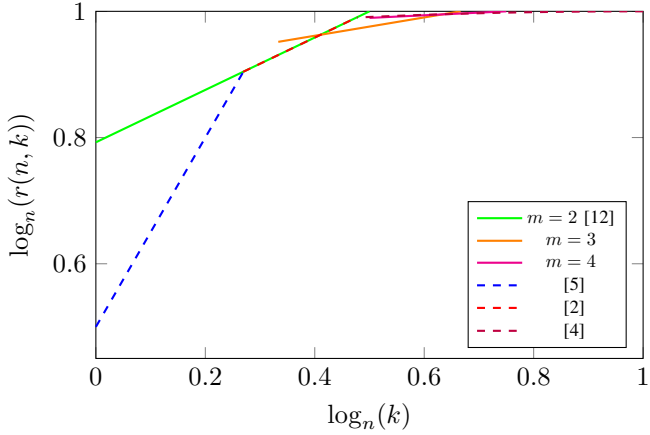
\begin{figure}
  \centering
  \begin{tikzpicture}
\pgfplotsset{compat = 1.3}
\begin{axis}[
	legend style={nodes={scale=0.7, transform shape}},
	width = \columnwidth,
	height = 0.7\columnwidth,
	xlabel = {$\log_n(k)$},
	xlabel style = {nodes={scale=0.8, transform shape}},
	ylabel = {$\log_n(r(n,k))$},
	ylabel style={nodes={scale=0.8, transform shape}},
	xmin = 0,
	xmax = 1,
	ymin = 0.45,
	ymax = 1,
	legend pos = south east]


\def\m{2}
\def\lambdaMax{3.0}
\addplot[color=green,
		domain = (\m-2)/\m:((\m-1)/\m),
		samples = 300,
		thick]
		{ln(\lambdaMax)/ln(2^\m)*(\m-1-\m*x)+\m*x-(\m-2)};
\addlegendentry{$m=2$ \cite{guo2013new}}

\def\m{3}
\def\lambdaMax{7.23606797749979}
\addplot[color=orange,
		domain = (\m-2)/\m:((\m-1)/\m),
		samples = 300,
		thick]
		{ln(\lambdaMax)/ln(2^\m)*(\m-1-\m*x)+\m*x-(\m-2)};
\addlegendentry{$m=3$}

\def\m{4}
\def\lambdaMax{15.543647468751944}
\addplot[color=magenta,
		domain = (\m-2)/\m:((\m-1)/\m),
		samples = 300,
		thick]
		{ln(\lambdaMax)/ln(2^\m)*(\m-1-\m*x)+\m*x-(\m-2)};
\addlegendentry{$m=4$}




\addplot[color=blue,
		domain = 0:0.269,
		samples = 300,
		dashed,
		thick]
		{1.5*x+0.5};
\addlegendentry{\cite{polyanskii2019constructions}}

\addplot[color=red,
		domain = 0.269:0.476,
		samples = 300,
		dashed,
		thick]
		{x*(2-ln(3)/ln(2))+ln(3)/ln(4)};
\addlegendentry{\cite{polyanskaya2019batch}}

\addplot[color=purple, mark=none,dashed,thick] table {upperBoundFile.txt};
\addlegendentry{\cite{asi2018nearly}}



\end{axis}
\end{tikzpicture}
  \caption{Comparison of parameters of binary batch codes based on $m$-variate lifts of the RS code for different values of $m$ with the upper bounds on the minimal redundancy of \cite{asi2018nearly,polyanskii2019constructions,polyanskaya2019batch}.}
  \label{fig:batch}
\end{figure}

\subsection{Our approach}
The main technique in this work is \emph{lifting} of codes, which was first studied in \cite{ben2011symmetric} in the context of LDPC codes and later employed to design locally correctable codes~\cite{guo2013new,guo2015high} and codes with the disjoint repair group property~\cite{li2019lifted,polyanskii2019lifted}.
Specifically, we construct batch codes from Reed-Solomon codes by lifting them to a higher dimension, while requiring the restriction of each codeword to a line to be a codeword of the RS code. This is shown~\cite{guo2013new} to be equivalent to generating a code by evaluating the polynomials in a vector space $\F_q^m$ from the linear span of all $m$-variate monomials, such that, when restricted to a line in the space, the resulting univariate polynomial is of  degree at most $d<q-1$.
An $m$-variate Reed-Muller (RM) code of order $d$ over a field $\F_q$ restricts the degree of the multivariate polynomials to be at most $d$ and thereby naturally provides this property. However, this causes the rate of the RM code to be very small. Lifted RS codes include not only the multivariate monomials of low degree, as RM codes do, but all polynomials which fulfill the required property. This gives a construction of codes with locality properties similar to RM codes, but of significantly higher rate. 
\subsection{Outline}
The remainder of the paper is organized as follows. In Section~\ref{ss::prelimiminaries}, we give rigorous definitions of lifted RS codes and batch codes and introduce several auxiliary notations. The rate of lifted RS codes can be determined by computing the fraction of so-called \emph{good monomials}, for which we will derive tight asymptotic formulas in Section~\ref{ss::lifted RS codes} and, thus, improve the result from~\cite{guo2013new}. In Section~\ref{ss::batchCodes}, we show that a lifted RS code is also an appropriate batch code, which gives us the best known upper bounds on the required redundancy $r(n,k)$ for  $k=n^\epsilon$ with $0.4< \epsilon < 0.6483$. We illustrate the trade-off between parameters of batch codes in Figure~\ref{fig:batch}.  Finally, we conclude with open problems in Section~\ref{ss::conclusion}.
\section{Preliminaries}\label{ss::prelimiminaries}
We start by introducing some notation that is used throughout the paper. Let $[n]$ be the set of integers from $1$ to $n$. A vector is denoted by bold lowercase letters such as $\d$. Let $q=2^{\ell}$ and $\F_q$ be a field of size $q$. We write $\log x$ to denote the logarithm of $x$ in base two. By $\Z_{\ge}$ and $\Z_{n}$  denote the set of non-negative integers and the set of integers from $0$ to $n-1$, respectively. In what follows, we fix $m$ to be a positive integer representing the number of variables. For $\d = (d_1,\dots, d_m)\in \Z_{q}^m$ and $\X=(X_1,\dots,X_m)$, let $\X^\d$ denote the monomial $\prod\limits{_{i=1}^m} X_i^{d_i}$ from $\F_q[\X]$. Let $\deg(\d)$ be the sum of components of $\d\in\Z_{\ge}^n$ and $|\d|$ be the number of non-zero components of $\d$. 

Let us define a partial order relation on $\Z_{q}$. We write $a\le_2 b$ if $a$ and $b$ can be represented by $a=\sum_{i=0}^{\ell-1} a^{(i)} 2^i$ and $b=\sum_{i=0}^{\ell-1} b^{(i)} 2^i$ with $a^{(i)},b^{(i)}\in\{0,1\}$ and $a^{(i)}\le b^{(i)}$ for all $i\in\{0,\dots, \ell-1\}$. We denote $a=(a^{(\ell-1)},...,a^{(0)})_2$.
For vectors $\d, \p\in \Z_{q}^m$, we write $\d\le_2 \p$ if $d_i\le_2 p_i$ for all $i\in[m]$. 

Define an operation $\Mods{q}$ that takes a non-negative integer and maps it to the element from $\Z_q$ as follows
$$
a \Mods{q}\eqdef
\begin{cases}
0,\,&\text{if   } a=0,\\
b\in [q-1],\,&\text{if }a\neq 0,\,a=b \Mod{q-1}.
\end{cases}
$$
It can be readily seen that if $a\,(\text{mod}^* q)=b$, then $T^{a} = T^{b} \Mod{T^q-T}$ in $\F_q[T]$.

For a function $f: \F_q^m \to \F_q$ and a set $S\subset \F^m_q$
let $f|_S$ denote the restriction of $f$ to the domain $S$. Abbreviate the set of all lines in $\F_q^m$ by
$$
\L_m\eqdef\left\{(\a T + \b)|_{T\in\F_q} \text{ for }\a,\b\in\F_q^m \right\}.
$$
We note that a multivariate polynomial restricted to a line is an univariate polynomial and the degree of the latter does not depend on the parameterization of the line. 

For a positive integer $d<q$, denote the set of univariate polynomials of degree less than $d$ by
$$
\cF_{d,q}\eqdef\{f(T)\in\F_q[T]:\,\,\deg(f)< d\}.
$$
\subsection{Lifted Reed-Solomon codes}
Let us recall the definition of lifted Reed-Solomon codes introduced in~\cite{guo2013new} in a more general form.
\begin{defn}[Lifted Reed-Solomon code, 
\cite{guo2013new}]\label{def::lifted RS code}
For an integer $m\ge 1$, the \textit{$m$-dimensional lift of the Reed-Solomon code} (or the \textit{$[m,d,q]$-lifted-RS code})  is the code
$$
\left\{(f(\a))|_{\a\in\F_q^m}:\,\,
\begin{aligned}
&f(\X)\in\F_q[\X]\text{ such that}\\
&\forall L\in \L_m:\,\,f|_L\in\cF_{d,q} 
\end{aligned}\right\}.
$$
\end{defn}
\begin{remark}
Note that the one-dimensional lift of a Reed-Solomon code represents the ordinary Reed-Solomon code of length $q$ and dimension $d$. Also, we observe that the $[m,d,q]$-lifted-RS code include all codewords of the $m$-variate RM code of order $d-1$ over $\F_q$.
\end{remark}
\begin{example}
Let $f(X_1,X_2)=X_1^2 X_2^2$. Then the $[2,3,4]$-lifted-RS code includes the codeword $\c=(f(a_1,a_2))|_{(a_1,a_2)\in\F_4^2}$ as for every line $L$, the degree of $f|_L$ is at most $2<3=d$. Indeed, given a line $L$ parameterized as $(\alpha_1 T +\beta_1, \alpha_2 T +\beta_2)|_{T\in \F_4}$ in $\F_4^2$, we have
\begin{align*}
f|_L&=f(\alpha_1 T +\beta_1, \alpha_2 T +\beta_2)=(\alpha_1 T +\beta_1)^2 (\alpha_2 T +\beta_2)^2\\
&\overset{(i)}{=}(\alpha_1^2 T^2 +\beta_1^2)(\alpha_2^2 T^2 +\beta_2^2)\\
&\overset{(ii)}{=} (\alpha_1^2\beta_2^2+\alpha_2^2\beta_1^2)T^2 
+\alpha_1^2\alpha_2^2T+\beta_1^2\beta_2^2,
\end{align*}
where in $(i)$ we used the property $2\alpha = 0$ for any $\alpha\in\F_4$, and $(ii)$ is implied by the fact that $T^4=T$ in $\F_4[T]$. On the other hand, the $2$-variate RM code of order $3$ doesn't contain $\c$ as the degree of $f$ is $4$, which is larger than $3$.
\end{example}
As shown in the example above, the characteristic of the field $\F_q$ can provide a gain in the number of good polynomials when compared with the RM code.

\begin{defn}[$d^*$-bad and good monomials] \label{def::bad * monomial}
Given a positive integer $d< q$, we say that a monomial $\X^\d$ with $\d\in\Z_q^m$ is \textit{$d^*$-bad} over $\F_q[\X]$ if there exists at least one $\i\in\Z_{q}^m$ such that $\i\le_2\d$ and $\deg(\i) \Mods{q} \in \{d,d+1,\dots,q-1\}$. A monomial is said to be \textit{$d^*$-good} if it is not $d^*$-bad.
\end{defn}
A characterization of lifting was established in~\cite{guo2013new}. We make use of this result for lifted Reed-Solomon codes. 
\begin{lemma}[Follows from~{\cite[Section 2]{guo2013new}}]\label{lem::lifted Reed-Solomon codes}
The $[m,q,d]$-lifted-RS code is equivalently defined as the evaluation of polynomials from the linear span of $d^*$-good monomials over $\F_q[\X]$. 
\end{lemma}
We do not include the proof of this lemma here but some elaboration on the connection between lifted-RS codes and $d^*$-good monomials is given in the Appendix.
Lemma~\ref{lem::lifted Reed-Solomon codes} suggests a way to compute the dimension of the $[m,q,d]$-lifted-RS code, namely one needs to estimate the size of the set of $d^*$-good $m$-variate monomials over $\F_q[\X]$. We carry out a careful analysis on the latter in Section~\ref{ss::lifted RS codes}.
\subsection{Batch codes}
We now proceed with a thorough definition of batch codes.
	\begin{defn}[Batch code, 
	\cite{ishai2004batch}]\label{def::batch code}
		Let $F:\,\F_q^n\to\F_q^N$ be a map that encodes a string $x_1,\dots,x_n$ to $y_1,\dots, y_N$ and $\C$ be the image of $F$.
		The code $\C$ will be called a \textit{$k$-batch code} if for every multiset of symbols $\{x_{i_1},\dots, x_{i_k} \},$ ${i_j}\in	[n]$, there exist $k$ mutually disjoint sets $R_{1},\dots , R_{k}\subset[N]$ (referred to as \textit{recovering sets}) and functions $g_1,\dots, g_k$ such that for all $\y\in\C$ and for all $j\in[k]$, $g_j(\y|_{R_{j}})=x_{i_j}$, where $\y|_{R}$ is the projection of $\y$ onto coordinates indexed by $R$.
	\end{defn}
	A one-way connection between lifted RS codes and batch codes is shown in Section~\ref{ss::batchCodes}.
\section{Code rate of lifted RS codes}\label{ss::lifted RS codes}
In this section, we investigate the code dimension of lifted RS codes. For this purpose, we first introduce the concept of $(q-r)$-bad monomials (slightly different from $(q-r)^*$-bad monomials) and derive an explicit evaluation formula to count the number of such monomials when the parameter $r\le m$ is fixed and the field size $q=2^\ell$ is scaled. Second, we show how to use the evaluation formula to derive a bound on the number of $(q-r)^*$-bad monomials for arbitrary $r\le q$. Our estimate improves upon the result presented in~\cite[Sections 3.2, 3.4]{guo2013new} for $m\ge 3$ and is consistent with the result for $m=3$ provided in~\cite{polyanskii2019lifted}.  
\subsection{Computing  the number of $(q-r)$-bad monomials}
Let us introduce a terminology useful for establishing the number of $d^*$-bad monomials. Let $r\le \min(m,q)$ be a fixed positive integer.
\begin{defn}[$(q-r)$-bad monomial]\label{def:: q - r bad}
We say that a monomial $\X^\d$ with $\d\in\Z_q^m$ is \textit{$(q-r)$-bad} over $\F_q[\X]$ if there exists at least one $\i\in\Z_q^m$ such that $\i\le_2\d$ and $\deg(\i) \pmod{q}=(q-r)$. 
\end{defn}
\begin{remark}
The difference with Definition~\ref{def::bad * monomial} is, roughly speaking, in the modulo operation, namely $\pmod{q}$  is used in Definition~\ref{def:: q - r bad}, whereas $\pmod{q-1}$ is used in Definition~\ref{def::bad * monomial}.
\end{remark}
Let $S_j(\ell)$ denote the set of tuples $\d\in\Z_q^m$, $q=2^\ell$, for which there exists $\i\le_2\d$ with $\deg(\i)= (q-r) + jq = (2^{\ell}-r)+j2^{\ell}$ and $s_j(\ell)$ be the cardinality of $S_j(\ell)$. We note that $S_j(\ell)$ also depends on $r$, however, we omit this in our notion as we fix $r$ and scale only $\ell=\log q$. We provide an evaluation formula that does not depend on $r$ as well. Clearly, $s_{j}(\ell)=0$ for $j\ge m$ as the maximal $\deg(\i)$ over admissible $\i$ is $m(q-1)$ which is smaller than $(q-r)+mq$. Therefore, we aim to compute $\sum_{i=0}^{m-1}s_i(\ell)$ since the number of $(q-r)$-bad monomials over $\F_q$ is bounded by this value from one side and by $s_0(\ell)$ from the other side.

\begin{example}
For $q=4$, $r=1$ and $m=2$ the set $S_0(2)$ is 
\begin{align*}
\begin{smallmatrix}
    S_0(2) &= & \{ &(3, 0),&(2, 1),&(3, 1),&(1, 2),&(3, 2),&(0, 3),&(1, 3),&(2, 3),&(3, 3)& \} \\
    &&& \downarrow \hphantom{,} & \downarrow \hphantom{,}&\downarrow \hphantom{,}&\downarrow \hphantom{,}&\downarrow \hphantom{,}&\downarrow\hphantom{,} &\downarrow \hphantom{,}&\downarrow\hphantom{,} &\downarrow\hphantom{,} & \\
     &\mathbf{i} \ :& &(3, 0)\hphantom{,} & (2, 1)\hphantom{,} & (3, 0)\hphantom{,} &(1, 2)\hphantom{,} & (3, 0)\hphantom{,} & (0, 3)\hphantom{,} & (1, 2)\hphantom{,} & (2, 1)\hphantom{,} & (3, 0)& 
\end{smallmatrix} \ .
\end{align*}
It is easy to check that for any $\mathbf{d} \in S_0(2)$ and the corresponding $\mathbf{i}$ it holds that $\mathbf{i} \leq_2 \mathbf{d}$ and $\deg(\mathbf{i}) = (q-r) + jq = 3$. The cardinality of the set is $s_0(2) = |S_0(2)| = 9$. For these parameters the only $\mathbf{d}$ with $\deg(\mathbf{d}) \geq q-r=3$ that is not $(q-r)$-bad is $\mathbf{d}=(2,2)$.

\end{example}

Let $\binom{b}{\ge a}$ denote the number of ways to choose an (unordered) subset of at least $a$ elements from a fixed set of $b$ elements. For $a<0$ or $a>b$, we assume that $\binom{b}{a}=0$.
\begin{proposition}\label{pr::recurrent formula}
The system of recurrence relations
$$
\begin{pmatrix}
s_0(\ell+1)\\
s_1(\ell+1)\\
\vdots\\
s_j(\ell+1)\\
\vdots\\
s_{m-1}(\ell+1)
\end{pmatrix} = 
 A_m
\begin{pmatrix}
s_0(\ell)\\
s_1(\ell)\\
\vdots\\
s_j(\ell)\\
\vdots\\
s_{m-1}(\ell)
\end{pmatrix}
$$
holds true, where the square $m\times m$ matrix $A_m$ is given by
$$
A_m \eqdef\left(\begin{smallmatrix}
\binom{m}{\ge 1} & \binom{m}{ 0} & 0 & 0 & \dots & 0 \\
\binom{m}{\ge 3} & \binom{m}{2} &\binom{m}{1} & \binom{m}{0} & \dots & 0 \\
\vdots & \vdots & \vdots & \vdots & \ddots & \vdots \\
\binom{m}{\ge 2j+1} & \binom{m}{ 2j} & \binom{m}{ 2j-1} & \binom{m}{ 2j-2} & \dots & \binom{m}{ 2j-m+2}
 \\
\vdots & \vdots & \vdots & \vdots & \ddots & \vdots \\
\binom{m}{\ge 2m-1} & \binom{m}{2m-2} & \binom{m}{ 2m-3} & \binom{m}{ 2m-4} & \dots & \binom{m}{ m}
\end{smallmatrix}\right).
$$
\end{proposition}
\begin{remark}
The proof of this technical statement can be found in the Appendix. 
 As a side note, this expression agrees with similar formulas for $m=2$ and $m=3$ mentioned in~\cite{guo2013new} and~\cite{polyanskii2019lifted}, respectively.  
\end{remark}
\begin{defn}[Largest eigenvalue $\lambda_m$]
Let $A_m$ be as in Proposition~\ref{pr::recurrent formula} and $\Lambda$ be the set of its eigenvalues. We define $\lambda_m$ to be the largest element from $\Lambda$.
\end{defn}
It is well known that the eigenvalues of a matrix are upper and lower bounded by the largest and smallest sum of its rows or columns, respectively. It follows directly from the structure of $A_m$ that $2^{m-1}\leq \lambda_m \leq 2^m$.
For the readers convenience, we provide $\lambda_m$ and $m-\log\lambda_m$ for $2\leq m \leq 9$ in Table~\ref{tab::eigenvalues}.

\begin{table}
    \centering
    \caption{The largest eigenvalue $\lambda_m$ of $A_m$, the resulting convergence rate $m-\log(\lambda_m)$ derived in this work, and the convergence rate $p_m$ of \cite{guo2013new} for different values of $m$.}
    \begin{tabular}{CCCC}
    m & \lambda_m & m-\log(\lambda_m) & p_m\\ \hline
2 & 3.0000 & 4.1504 \times 10^{-1}& 4.1504 \times 10^{-1}\\
3 & 7.2361 & 1.4479 \times 10^{-1}& 1.1360 \times 10^{-2}\\
4 & 15.5436 & 4.1747 \times 10^{-2} & 2.8233 \times 10^{-3}\\
5 & 31.7877 & 9.6043 \times 10^{-3} & 4.6986 \times 10^{-4}\\
6 & 63.9217 & 1.7653 \times 10^{-3}& 1.1742 \times 10^{-4}\\
7 & 127.9763 & 2.6714\times 10^{-4}& 2.9353\times 10^{-5} \\
8 & 255.9939 & 3.4467 \times 10^{-5}& 2.8664 \times 10^{-8} \\
9 & 511.9986 & 3.8959 \times 10^{-6} & 2.6872 \times 10^{-9}
    \end{tabular}
    \label{tab::eigenvalues}
\end{table}

Note that the order of $s_j(\ell)$ is the maximum value in the matrix $A_m^\ell$, the $\ell$th power of $A_m$.  The exponential growth rate of the matrix powers $A_m^\ell$ as $\ell\to\infty$ is controlled by $\lambda_m^\ell$. Since all elements of $A_m^{m-1}$ are positive (except the $m$th row which has all zeros but the last entry), the matrix $A_m$ has only one eigenvalue of maximum modulus by Perron-Frobenius theorem for non-negative matrices (e.g., see~\cite[Theorem 8.5.2]{horn2012matrix}). Finally, we obtain the following statement.
\begin{corollary}\label{cor::number of bad monomials}
For an integer $r\le m$, the number of $(q-r)$-bad monomials is $\Theta(\lambda_m^{\ell})=\Theta(q^{\log \lambda_m})$ as $q\to\infty$.
\end{corollary}
\subsection{Computing  the number of $(q-r)^*$-bad monomials}
Now let $r\le q$ (the restriction $r\le m$ is no longer necessary, i.e., $r$ could be very large). By Definition~\ref{def::bad * monomial}, a monomial $\X^\d$ is $(q-r)^*$-bad if there exists an $\i\in\Z_{q}^m$ such that $\i\le_2 \d$ and $\deg(\i) \Mods{q} \in \{q-r,q-r+1,\dots,q-1\}$. The latter condition is equivalent to 
$$
\deg(\i)= q-r_0 + (q-1)j=(q-r_0-j) + qj
$$
for some $r_0\in[r]$ and $j\in \Z_{m}$. Let us drop the $\lceil\log(r+m)\rceil$ least significant bits in every component of $\d$ and $\i$ to obtain some $\d'$ and $\i'$ from $\Z_{q'}^m$ with $q'=2^{\ell'}$ and $\ell'=\ell-\lceil\log(r+m)\rceil$. Then we have that $\i'\le_2 \d'$ and 
$$
(q'-m)+jq'\le \deg(\i')\le \lfloor\deg(\i)/2^{\ell-\ell'}\rfloor\le (q'-1)+jq'.
$$
Therefore, by Definition~\ref{def:: q - r bad}, we have that $\X^{\d'}$ is $(q'-r')$-bad over $\F_{q'}[\X]$ for some positive integer $r'\le m$. By simple counting arguments and Corollary~\ref{cor::number of bad monomials}, the following statement is implied.
\begin{corollary}\label{cor:: number of bad * monomials}
For an integer $r<q=2^{\ell}$, the number of $(q-r)^*$-bad monomials is $\Theta(r^{m-\log\lambda_m} q^{\log \lambda_m})$ as $\ell\to\infty$.
\end{corollary}
\begin{proof}[Proof of Corollary~\ref{cor:: number of bad * monomials}]
The number of $(q-r)^*$-bad monomials can be bounded by the number of $(q'-r')$-bad monomials with $r'\le m$ multiplied by the number of ways to choose $m\lceil \log(r+m)\rceil$ bits. By Corollary~\ref{cor::number of bad monomials}, it can be estimated as 
$$
m2^{m}(r+m)^m O\left({q'}^{\log \lambda_m}\right)=O\left(r^{m-\log\lambda_m} q^{\log \lambda_m}\right),
$$
where the factor $m$ comes from the number of choices for the parameter $r'\in [m]$ and $2^m(r+m)^m\ge 2^{m\lceil\log(r+m)\rceil}$ is the number of ways to choose $m\lceil \log(r+m)\rceil$ bits.  

Now let us elaborate on showing that the number of $(q-r)^*$-bad monomials is $\Omega\left(r^{m-\log\lambda_m} q^{\log \lambda_m}\right)$. Take all $(q'-1)$-bad monomials $\X^{\d'}$ over $\F_{q'}[\X]$ with the property that there exists $\i'\le_2 \d'$ such that $\deg(\i')=q'-1$. 
By Proposition~\ref{pr::recurrent formula} and Corollary~\ref{cor::number of bad monomials}, the number of such monomials can be bounded as $\Omega(q'^{\log \lambda_m})$. Define
$$
\ell_0:=\lceil\log(m+r)\rceil-\lfloor \log r\rfloor.
$$
Then we concatenate every component $d'_j$ of $\d'=(d'_1,\ldots,d'_m)$  with the all-one string of length $\ell_0$
and an arbitrary binary string of length $\lfloor \log r\rfloor$. The total number of obtained tuples $\d\in \Z_q^m$ is then
$$
2^{m\lfloor \log r \rfloor}\Omega\left(q'^{\log \lambda_m}\right)=\Omega\left(r^{m-\log\lambda_m} q^{\log \lambda_m}\right).
$$
For every resulting tuple $\d$, the monomial $\X^\d$ is also $(q-r)^*$-bad over $\F_q[\X]$. Indeed, we can construct an appropriate $\i$ based on $\i'$. To see this, we concatenate every component $i'_j$ (except $i'_1$) with the all-zero string of length $\lceil\log(r+m)\rceil$, and $i'_1$ with the all-one string of length  $\ell_0$ and the all-zero string of length $\lfloor \log r\rfloor$.
Then we have $\i\le_2\d$ and $\deg(\i)$ can be easily bounded as $q-r\le\deg(\i)\le q -1$.
This completes the proof.
\end{proof}

\begin{example}
Consider the parameters $q' =2^{\ell'} = 4$, $m=2$, $r=2$, and $q = 2^{\ell' + \lceil\log(r+m)\rceil} = 16$. As shown in the previous example, we have $\mathbf{d}' = (1,3) \in S_0(\ell')$ with $\mathbf{i}' \leq_2 \mathbf{d}'$ for $\mathbf{i}' = (1,2)$. The binary representations of $\mathbf{d}'$ and $\mathbf{i}'$ are given by 
\begin{align*}
    \mathbf{d}' &= (01,11)_2\\
    \mathbf{i}' &= (01,10)_2
\end{align*}
Concatenating the all-one string of length $\ell_0 =\lceil\log(m+r)\rceil-\lfloor \log r\rfloor  =1$ followed by arbitrary strings of length $\lfloor \log r\rfloor = 1$ to the components of $\mathbf{d}'$ gives the tuples
\begin{align*}
    \mathbf{d}_1 &= (0110,1110)_2\\
    \mathbf{d}_2 &= (0110,1111)_2\\
    \mathbf{d}_3 &= (0111,1110)_2\\
    \mathbf{d}_4 &= (0111,1111)_2 \ .
\end{align*}
The $\mathbf{i}$ such that $\mathbf{i}\leq \mathbf{d}_j, \ j=1,2,3,4$ can be found by concatenating every component $i_j'$ except for $i_1'$ with $\lceil \log(r+m) \rceil= 2 $ zeros and $i_1$ with $\ell_0 = 1$ one and $\lfloor \log r\rfloor = 1$ zero, to obtain
\begin{align*}
    \mathbf{i} = (0110,1000)_2 \ .
\end{align*}
The degree of $\mathbf{i}$ is $\deg(\mathbf{i}) = 14 \geq q-r$.
\end{example}

\subsection{Code rate of lifted RS codes}
\begin{theorem}\label{th:: number of bad monomials}
The rate of the $[m,q,q-r]$-lifted-RS code is
$$
1-\Theta\left((q/r)^{\log \lambda_m-m}\right)\quad \text{as } q\to\infty.
$$
\end{theorem}
It is clear that the rate approaches $1$ for $r=o(q)$ as $\lambda_m<2^m$. This fact was also proved in~\cite{guo2013new} in order to show the existence of high rate locally correctable codes with sublinear locality. Let us illustrate the improvement of Theorem~\ref{th:: number of bad monomials} compared to the result from~\cite{guo2013new}. We take $r=O(1)$ and check that the convergence rate of our estimate is $1-\Theta\left(q^{\log \lambda_m - m}\right)$. 
The arguments from~\cite{guo2013new} show that for $m\ge 2$, the rate is
$$
1- O\left(\left(1-2^{-m\lceil\log m\rceil}\right)^{\log q/\lceil\log m\rceil}\right)= 1-O(q^{-p_{m}}),
$$
where $p_m\eqdef -\log\left(1-2^{-m\lceil\log m\rceil}\right)/\lceil\log m\rceil$. In Table~\ref{tab::eigenvalues}, we depict some values of $m-\log \lambda_m$ and $p_m$ for $2\le m\le 9$.
\begin{proof}[Proof of Theorem~\ref{th:: number of bad monomials}]
Lemma~\ref{lem::lifted Reed-Solomon codes} provides a way to estimate the code rate of $[m,q-r,q]$-lifted-RS codes by computing the fraction of $(q-r)^*$-good monomials. By Corollary~\ref{cor:: number of bad * monomials}, the rate is
$$
1 - \Theta\left(r^{m-\log\lambda_m} q^{\log \lambda_m}\right)/q^m=1-\Theta\left((q/r)^{\log \lambda_m-m}\right)
$$
as $q\to\infty$. This completes the proof.
\end{proof}

\section{Batch codes based on lifted RS codes}\label{ss::batchCodes}
In this section, a new construction of binary batch codes is presented. To this end, we first provide a construction of non-binary $k$-batch codes of length $n$ based on the $m$-dimensional lifts of an RS code. After that,  we compute the parameters of this construction in the asymptotic regime for the availability parameter $k=n^{\varepsilon}$ with real $\varepsilon\in[\frac{m-2}{m}, \frac{m-1}{m}]$. Finally, we show how to convert this construction into a binary batch code.
\begin{theorem}\label{th:: nonbinary batch code}
Fix integers $q$, $m$ and $r<q$. The $[m,q-r,q]$-lifted-RS code has the following properties:
		\begin{enumerate}
			\item The length of the code is $q^m$.
			\item The rate of the code is $1-\Theta\left((q/r)^{\log \lambda_m-m}\right)$ as $q\to\infty$.
			\item The code is a $k$-batch code for $k=q^{m-2}r$.
		\end{enumerate}
\end{theorem}
\begin{proof}[Proof of Theorem~\ref{th:: nonbinary batch code}]
The first property follows from Definition~\ref{def::lifted RS code}. The second property is implied by Theorem~\ref{th:: number of bad monomials}. 

To prove the third property, we first note that a lifted RS code is a linear code over $\F_q$ and it can be encoded systematically. Let $\y$ be a codeword of the $[m,d,q]$-lifted-RS code. Since every coordinate of $\y$ is simply the evaluation $f(\a)$ for some $\a\in\F_q^m$, we can index coordinates of our code by elements $\a$ from $\F_q^m$. 

Now we shall prove a slightly stricter condition than required for $k$-batch codes, namely for every multiset of symbols $\{y_{\a_1},\dots,y_{\a_k}\}$, there exist mutually disjoint sets $R_1,\dots, R_k\subset \F_q^m$ and some functions $g_1,\dots,g_k$ such that $y_{\a_i}=g_i(\y|_{R_i})$. 
Let us prove the existence of $R_1,\dots,R_k$ by using the inductive procedure described below. 

To reconstruct $y_{\a_1}$, we take an arbitrary line $L_1$ in $\F_q^m$ containing $\a_1$ and let $R_1=L_1\setminus \{\a_1\}$. As the restriction of polynomial $f$ to a line $L_1$ has degree less than $q-r$ by definition of lifted RS codes, we can interpolate $f|_{L_1}$ by reading evaluations of $f$ at some $q-r$ points on the line $L_1$ and evaluate $f|_{L_1}$ at point $\a_1$. Suppose that for $k'<k$, symbols $\{y_{\a_1},\dots,y_{\a_{k'}}\}$ can be reconstructed by using recovering sets $R_1,\dots, R_{k'}$, where $R_i$ is a subset of a line $L_i$ from the space $\F_q^m$. Since the number of lines passing through the point $\a_{k'+1}$ is larger than $q^{m-1}$ and the total number of points already employed for recovering $\{y_{\a_1},\dots,y_{\a_{k'}}\}$ is at most $qk'$, we conclude that there exists a line $L_{k'+1}$ among $q^{m-1}$ ones such that the cardinality of the intersection
$$
\left|L_{k'+1} \bigcap \left\{\bigcup_{i\in[k']}L_{i}\right\}\right|\le \frac{qk'}{q^{m-1}}< \frac{qk}{q^{m-1}}=r.
$$
Therefore, we can reconstruct $\y_{\a_{k'+1}}$ by reading evaluations of $f$ at some $q-r$ unused points on $L_{k'+1}$, interpolating the univariate polynomial $f|_{L_{k'+1}}$ of degree less than $q-r$ and evaluating the latter at point $\a_{k'+1}$. 

Thus, the required multiset of codeword symbols can be determined by this procedure. This completes the proof.
\end{proof}
In the next statement we show a connection between parameters of the non-binary batch code constructed in Theorem~\ref{th:: nonbinary batch code}.
\begin{theorem}\label{th::asymptotic non-binary batch code}
Given a positive integer $m$, for any real $\varepsilon$ with $\frac{m-2}{m}\le\varepsilon <\frac{m-1}{m}$ and a power of two $q$,
there exists a $n^\varepsilon$-batch code of length $N=q^m$ and dimension $n$ over $\F_q$
such that the redundancy, $N-n$,   satisfies
$$
N-n = O\left(n^{(m-\log \lambda_m)\varepsilon + ((m-1)\log \lambda_m/m - m + 2)}\right).
$$
\end{theorem}
\begin{proof}[Proof of Theorem~\ref{th::asymptotic non-binary batch code}]
Let $r= \lceil q^{m\varepsilon-m+2}\rceil \ge n^{\varepsilon-(m-2)/m}$. By Theorem~\ref{th:: nonbinary batch code}, there exists a $k$-batch code with $k=rq^{m-2}\ge q^{m\varepsilon}=n^\varepsilon$ over $\F_q$ of length $N=q^m$ and redundancy at most
\begin{align*}
N-n&=O\left(r^m \lambda_m^{\ell-\log r}\right)\\
&=O\left(2^{\ell m (m\varepsilon-m+2)}\lambda_m^{\ell-\ell (m\varepsilon-m+2)}\right)\\
&=O\left(n^{(m-\log \lambda_m)\varepsilon + ((m-1)\log \lambda_m/m - m + 2)}\right) .
\end{align*}
\end{proof}
\begin{theorem}\label{th::binary batch code}
	Given a positive integer $m$, for any real $\varepsilon$ with $\frac{m-2}{m}\le  \varepsilon \le\frac{m-1}{m}$, any real $\delta > 0$ and an integer $n$ sufficiently large,
there exists a binary $n^{\varepsilon-\delta}$-batch code of length $N$ and dimension $n$
such that the redundancy, $N-n$, satisfies
$$
N-n = O\left(n^{(m-\log \lambda_m)\varepsilon + ((m-1)\log \lambda_m/m - m + 2)}\right).
$$
\end{theorem}
\begin{proof}[Proof of Theorem~\ref{th::binary batch code}]
	Let $\C$ be a non-binary batch code from Theorem~\ref{th::asymptotic non-binary batch code}. We construct the binary batch code $\C'$ from $\C$ by converting each symbol of the alphabet of size $q$ to $\log q=\log N^{1/m}=\frac{1}{m}\log N=\Theta(\log n)$ bits. Denote the length, dimension of the binary
	code by $N', n'$ respectively. Thus, $n'= \Theta(n \log n)$ and $N'=\Theta(N \log n)$. Therefore, $n=\Theta(n'/\log n')$. Denote by $r' = N'-n'$ the redundancy of the binary code and by $k'$ be the availability parameter of the new code. 
	
	First, we note that the availability parameter of $\C'$ is at least that of $\C$. Indeed, we know that each bit in $\C'$ is a bit among $\log q$ bits representing some symbol in $\C$. For each recovering set of a symbol in $\C$, we have the corresponding recovering set for any bit from the image of this symbol in $\C'$. Therefore, $k'\ge k=n^\epsilon\ge (n'/\log n')^{\varepsilon}$.
	
	Second, we rewrite the redundancy $r'$ in terms of $n'$ as
	\begin{align*}
	r'&=N'-n'=O((N-n)\log n) \\
	&=O\left(n'^{(m-\log \lambda_m)\varepsilon + ((m-1)\log \lambda_m/m - m + 2)}\log n'\right) .
	\end{align*}
As for any $\delta>0$ and sufficiently large $n$ we have $\log n < n^{\delta}$, the required statement is proved.
\end{proof}
\section{Conclusion}\label{ss::conclusion}
In this paper, we have investigated the code rate of lifted Reed-Solomon codes and discussed how to use the latter to construct batch codes. Our results are two-fold.
\begin{enumerate}[wide]
    \item We have improved the estimate on the rate of the $m$-dimensional lifts of the RS codes when the field size is large. In particular, we have shown that 
for $r=O(1)$, the $[m,q-r,q]$-lifted-RS code has rate $1-\Theta(q^{\log{\lambda_m}-m})$ as $q\to\infty$. As a further research direction, it would be of great interest to analyze lifted multiplicity codes when the parameter of lifting $m\ge 3$. This would continue the study initiated by Li and Wootters in~\cite{li2019lifted} of two-dimensional lifts. It has been shown that this natural generalization makes the construction much more flexible for various parameters.
\item The locality property of lifted RS codes makes them attractive for constructing locally correctable codes and codes with the disjoint repair group property. Additionally, we have shown that a $[m,q-r,q]$-lifted-RS code is also a $k$-batch code with $k=rq^{m-2}$. This improves the known upper bounds on the redundancy of batch codes in some parameter regimes. On the other hand, there is no lower bound on the redundancy beyond the lower bound for $k=3$, stating~\cite{wootters2016linear} that the redundancy of linear batch codes of length $N$ is $\Omega(\sqrt{N})$.  An improvement of the latter for larger $k$ remains an interesting open problem.
\end{enumerate}
\appendix
\subsection{Lifted-RS codes from $d^*$-good monomials}
We shall show that lifted-RS codes include the evaluation of $d^*$-good monomials (and their linear combinations). By Definition~\ref{def::lifted RS code}, it suffices that every $d^*$-good monomial $f(\X)=\X^\d$ over $\F_q$ satisfies the property that for any line $L\in\L_m$, the restriction $f|_L$ is an univariate polynomial of degree less than $d$.  Let a line $L$ be parameterized as $(\a T+\b)|_{T\in\F_q}$ and $\0$ be the all-zero vector. Then, we have that 
\begin{align*}
f|_L &=(\a T+\b)^\d \\
&= \sum_{\0\le \i\le \d}\prod_{j=1}^{m}a_j^{i_j}b_j^{d_j-i_j}\binom{d_j}{i_j} T^{i_j}\\
&=\sum_{k=0}^{q-1} c_k T^k,
\end{align*}
where coefficients $c_k$ are derived by using the property $T^q=T$ in $\F_q[T]$
$$
c_k\eqdef \sum_{\substack{\0\le \i\le \d \\ \deg(\i)\Mods{q} = k }}\prod_{j=1}^{m}a_j^{i_j}b_j^{d_j-i_j}\binom{d_j}{i_j}.
$$
By Definition~\ref{def::bad * monomial}, for $k\ge d$, there is no $\i\in\Z_q^m$ such that $\i\le_2\d$ and $\deg(\i)\Mods{q} = k$. Thus, for $k\ge d$ and every $\i$ used in the summation above for defining $c_k$, there exists some coordinate $j\in[m]$ such that $i_j\not\le_2 d_j$. By Lucas's Theorem (e.g., see~\cite{guo2013new,li2019lifted}), for integers $d_j=(d_j^{(\ell-1)},...,d_j^{(0)})_2$ and $i_j=(i_j^{(\ell-1)},...,i_j^{(0)})_2$ it holds that 
\begin{equation*}
    \binom{d_j}{i_j} = \prod_{\xi=0}^{\ell-1} \binom{d_j^{(\xi)}}{i_j^{(\xi)}} \mod 2  .
\end{equation*}
It follows that if $i_j\not\le_2 d_j$ the coefficient $\binom{d_j}{i_j}=0$ in $\F_q$ (as $q$ is a power of two) and therefore $c_k=0$ for all $k\ge d$.

We have proved that the restriction of $\X^\d$ to any line is an univariate polynomial of degree at most $d-1$. Therefore, the $[m,d,q]$-lifted-RS code includes codewords 
$$
\{(\a^\d)|_{\a\in\F_q^m}:\quad \X^\d \text{ is }d^*\text{-good over }\F_q[\X]\}
$$
and their linear combinations over $\F_q$. This completes the proof.
\subsection{Proof of Proposition~\ref{pr::recurrent formula}}\label{ss::proof of recurrent formula}
We have 
two important ingredients, Lemma~\ref{lem::decrease weight} and Lemma~\ref{lem::drop the leading bit}, in the proof of Proposition~\ref{pr::recurrent formula}.
\begin{lemma}\label{lem::decrease weight}
    If $\d\in S_{j}(\ell)$ for a non-negative integer $j$, then $\d\in S_{l}(\ell)$ for any non-negative integer $l< j$.
\end{lemma}
\begin{proof}[Proof of Lemma~\ref{lem::decrease weight}]
As $\d\in S_j(\ell)$, there exists some $\i$ such that $\i\le_2 \d$ and $\deg(\i)=(q-r)+jq=(2^{\ell}-r)+j 2^{\ell}$. 
We shall prove that there exists $\i'$ such that $\i'\le_2 \i$ and $\deg(\i')=(2^{\ell}-r)+l 2^{\ell}$. This is sufficient for showing $\d\in S_l(\ell)$. To see it, we provide an iterative procedure that takes an arbitrary $\i\in \Z_q^m$ with $\deg(\i)\ge j2^\ell$ and outputs $\a\le_2 \i$ with $\deg(\a) = \deg(\i) - (j-l)2^{\ell}$ for $l\in[j]$. The procedure goes from the leading bits to the least significant ones and replaces some ones in the binary representations of $\i=(i_1,\ldots,i_m)$ by zeros.
\begin{enumerate}[labelwidth=!]
    \item \textbf{Step 1.} Let us initialize $\a\gets \i$ and $\Delta \gets (j-l)$ and $h \gets {\ell}$.
    \item \textbf{Step 2.} If $h=0$, output $\a$. Else, let $h \gets h-1$ and $\Delta\gets 2\Delta$. Compute $\delta=\Delta-\sum_{\xi=1}^m a_\xi^{(h)}$. 
    If $\delta > 0$, let $\Delta\gets \Delta-\delta$ and $a_\xi^{(h)}\gets 0$ for all $\xi\in [m]$. Repeat Step 2.     Else, let  $m'$ satisfy $\Delta-\sum_{\xi=1}^{m'} a_{\xi}^{(h)}=0$ and let $a_{\xi}^{(h)}\gets 0$ for all $\xi \in [m']$. Output $\a$.
\end{enumerate}
According to the procedure, we output the correct $\a$ if we do the else-part in Step 2 at some point. Assume this never happens. This means that we output the all-zero tuple at the end. However,  $\Delta=(j-l)2^\ell - \deg(\i)>0$ at the final step which contradicts with $\deg(\i)\ge j2^\ell$.
This completes the proof.
\end{proof}

\begin{example}
Consider the parameters $q=2^\ell=4$, $m=2$, $r=2$, $j=1$, and $l=0$. For the element $\mathbf{d} = (3,3) \in S_1(2)$ and $\mathbf{i} = (3,3) = (11,11)_2$ with $\mathbf{i}\leq_2 \mathbf{d}$ we will find the corresponding $\mathbf{a}$ with $\mathbf{a} \leq_2 \mathbf{i}$ and $\deg(\mathbf{a}) = \deg(\mathbf{i}) -(j-l)2^\ell = 2$. 
\begin{enumerate}
    \item \textbf{Step 1.} Initialize $\mathbf{a} \gets (3,3)$ and $\Delta \gets j-l = 1$ and $h\gets \ell=2$.
    \item \textbf{Step 2.} Let $h\gets h-1 = 1$ and $\Delta \gets 2\Delta=2$. Compute $\delta = \Delta - \sum_{\xi =1}^m a_\xi^{(h)} = 0$. Since $\delta \not> 0$ we choose $m'=2$ to satisfy $\Delta-\sum_{\xi=1}^{m'} a_\xi^{(h)} = 0$ and set $a_{1}^{(1)}\gets 0$, $ a_{2}^{(1)} \gets 0$ to obtain $\mathbf{a} = (01,01)_2 = (1,1)$.
\end{enumerate}
As $\mathbf{a} \leq_2 \mathbf{i} \leq_2 \mathbf{d}$ and $\deg(\mathbf{a}) = q-r= 2$ it follows that $\mathbf{d} \in S_0(2)$.
\end{example}

Let us introduce some auxiliary functions. We define two maps $F_{drop}:\Z_{2^\ell}\to \Z_{2^{\ell-1}}$ and $F_{lead}:\Z_{2^\ell}\to \Z_{2}$ that take an integer $a=\sum_{i=0}^{\ell-1} a^{(i)} 2^i$ and output $a-2^{\ell-1}a^{(\ell-1)}$ and $a^{(\ell-1)}$, respectively (we either drop the leading bit in the binary representation of $a$ or output it). 
We extend the maps $F_{drop}$ and $F_{lead}$ to $\Z_{2^\ell}^m$ in a straightforward manner by applying functions to each component of a vector $\a\in\Z_{2^\ell}^m$, that is
\begin{align*}
&F_{drop}(\a)=(F_{drop}(a_1),\dots,F_{drop}(a_m)),\\ &F_{lead}(\a)=(F_{lead}(a_1),\dots,F_{lead}(a_m)).
\end{align*}
For an integer $a$, we denote $\max(a,0)$ by $(a)^+$.
\begin{lemma}\label{lem::drop the leading bit}
If $\d\in S_{j}(\ell+1)$ for a non-negative integer $j$, then $F_{drop}(\d)$ belongs to $S_{0}(\ell), S_{1}(\ell),\dots,S_{(2j+1-|F_{lead}(\d)|)^{+}}(\ell)$.
\end{lemma}
\begin{proof}[Proof of Lemma~\ref{lem::drop the leading bit}] By definition, if $\d\in S_{j}(\ell+1)$, then there exists some $\i\in\Z_{2^{\ell+1}}^m$ with $\i\le_2\d$ and $\deg(\i)=(2^{\ell+1}-r)+j2^{\ell+1}$. It is obvious that if the leading bits in $\i$ are dropped, then the sum of components of $F_{drop}(\i)$
\begin{align*}
\deg(F_{drop}(\i)) &= \deg(\i)-|F_{lead}(\i)|2^{\ell}\\
&= (2^{\ell}-r)+(2j+1-|F_{lead}(\i)|)|2^{\ell}.
\end{align*}
 Since we also have the property $F_{drop}(\i)\le_2 F_{drop}(\d)$, we obtain that $F_{drop}(\d)$ belongs to $S_{2j+1-|F_{lead}(\i)|}(\ell-1)$. Additionally, we note that $|F_{lead}(\i)|\le \min(2j+1,|F_{lead}(\d)|)$ as $\i\le_2\d$ and $\deg(\i)=(2^{\ell}-r)+j2^{\ell}$. From this and Lemma~\ref{lem::decrease weight}, we conclude that $r(\d)$ belongs to $S_{0}(\ell-1)$, $S_{1}(\ell-1)$, $\dots$, $S_{(2j+1-|F_{lead}(\d)|)^{+}}(\ell-1)$. This completes the proof.
\end{proof}
Note that we can uniquely encode $\d\in\Z_{2^{\ell+1}}^m$ by the pair $(F_{lead}(\d), F_{drop}(\d))$. Let us define the set $Pair(j)$ as follows
$$
Pair(j)=\left\{(F_{lead}(\d), F_{drop}(\d)):\quad \d\in S_j(\ell+1) \right\}.
$$
For $w\in\{0,\dots,m\}$, we define the set $T^{(w)}(j)$ as follows
$$
T^{(w)}(j) = 
\{(\v,\y):\,\, \v\in\Z_2^m,\y\in S_{(2j+1-w)^+}(\ell),\, |\v|=w \}.
$$
Recall that $s_j(\ell)=|S_j(\ell)|$. To show 
\begin{align*}
s_j(\ell+1) &= {m \choose \ge 2j+1}s_0(\ell) +{m \choose 2j}s_1(\ell) \\
&+ {m \choose  2j-1}s_2(\ell)+\dots+{m \choose 2j-m+3}s_{m-2}(\ell)\\
&+{m \choose 2j-m+2}s_{m-1}(\ell)+{m \choose 2j-m+1}s_{m}(\ell) ,
\end{align*}
it remains to prove that the disjoint union of $T^{(w)}(j)$ coincides with $Pair(j)$, that is
$$
\bigsqcup_{w\in\{0,\dots,m\}} T^{(w)}(j)=Pair(j).
$$

First, we prove that each element in $Pair(j)$ is covered by the union. Let $(F_{lead}(\d), F_{drop}(\d))\in Pair(j)$ for some $\d\in S_j(\ell+1)$. By denoting $w=|F_{lead}(\d)|$ and applying Lemma~\ref{lem::drop the leading bit}, we get that $F_{drop}(\d)\in S_{(2j+1-w)^+}(\ell)$. Therefore, $(F_{lead}(\d), F_{drop}(\d))\in T^{(w)}(j)$. 

Second, we show that each element in $T^{(w)}(j)$ is included in $Pair(j)$. Let $(\v,\y)\in T^{(w)}(j)$. Construct $\d\in \Z_{2^{\ell+1}}^m$ to satisfy $F_{lead}(\d)=\v$ and $F_{drop}(\d)=\y$. By definition, we have that $|\v|=w$ and $\y\in S_{(2j+1-w)^{+}}(\ell)$. The latter means that there exists an $\i$ such that $\i\le_2 \y$ and $\deg(\i)=(2^{\ell}-r) + (2j+1-w)^{+}2^{\ell}$. Construct $\i'\in \Z_{2^{\ell+1}}^m$ such that $F_{drop}(\i')=\i\le_2 \y=F_{drop}(\d)$ and $F_{lead}(\i')\le_2 \v=F_{lead}(\d)$ and $|F_{lead}(\i')|=\min(2j+1,w)$. Thus, we obtain that $\i'\le_2 \d$ and $\deg(\i')=(2^{\ell+1}-r)+j2^{\ell+1}$. This completes the proof.

	\bibliographystyle{IEEEtran}
	\bibliography{lifted}
\end{document}